\documentclass[12pt]{article}  
\usepackage{amssymb}
\usepackage{amsmath}
\usepackage{times}
\usepackage{fullpage}
\usepackage{epsfig}
\usepackage{graphicx}
\usepackage{epstopdf}
\usepackage{todonotes}
\usepackage{hyperref}

\newcommand{\E}{{\sf E}}
\newcommand{\Var}{{\sf Var}}
\newcommand{\ttl}{{\tt 1}}
\newcommand{\tto}{{\tt 0}}

\newcommand{\drop}[1]{}
\newcommand{\ppmod}{\rule{-1.5ex}{0ex}\pmod}

\newcommand{\floor}[1]{\lfloor {#1} \rfloor}

\newcommand{\req}[1]{(\ref{#1})}

\newtheorem {lemma} {Lemma}[section]
\newtheorem {fact} [lemma] {Fact}
\newtheorem {definition} {Definition}

\newtheorem {theorem}[lemma] {Theorem}
\newtheorem {observation}[lemma] {Observation}
\newtheorem {question}{Question}[section]
\newtheorem {exercise}[question]{Exercise}
\newcommand{\qed}{\rule{1ex}{1ex}}
\newenvironment{proof}[1][]{\paragraph*{Proof{#1}}}{\hfill \qed\smallskip\\}

\newcommand\eps\varepsilon
\newcommand\fct\rightarrow

\newcommand\ceil[1]{\lceil {#1}\rceil}

\newcommand\ol\overline

\usepackage{color}
\usepackage{xcolor}

\title{High Speed Hashing for Integers and Strings} 
\author{Mikkel Thorup}
\begin{document}
\maketitle 
\begin{abstract}
These notes describe the most efficient hash functions currently known
for hashing integers and strings. These modern hash functions are
often an order of magnitude faster than those presented in standard text
books. They are also simpler to implement, and hence a clear win in
practice, but their analysis is harder. Some of the most practical
hash functions have only appeared in theory papers, and some of them
require combining results from different theory papers. The goal here
is to combine the information in lecture-style notes that can be used
by theoreticians and practitioners alike, thus making these practical
fruits of theory more widely accessible.
\end{abstract}

\section{Hash functions}
The concept of truly independent hash functions is extremely useful in
the design of randomized algorithms. We have a large universe $U$ of
keys, e.g., 64-bit numbers, that we wish to map randomly to a range
$[m]=\{0,\ldots,m-1\}$ of hash values.  A {\em truly random hash
  function\/} $h:U\fct [m]$ assigns an independent uniformly random
variable $h(x)$ to each key in $x$. The function $h$ is thus a
$|U|$-dimensional random variable, picked uniformly at random among
all functions from $U$ to $[m]$.

Unfortunately truly random hash functions are idealized objects that
cannot be implemented. More precisely, to represent a truly random
hash function, we need to store at least $|U|\log_2 m$ bits, and in
most applications of hash functions, the whole point in hashing is
that the universe is much too large for such a representation 
(at least not in fast internal memory). 

The idea is to let hash functions contain only a small element or seed 
of randomness so that the hash function is sufficiently random for the 
desired application, yet so that the seed is small enough that we can store it when first it is fixed. As an example, if $p$ is prime, a random
hash function $h:[p]\to[p]=\{0,\ldots,p-1\}$ could be 
$h(x)=(ax+b)\bmod p$  where $a$ and $b$ are random numbers that together form
the random seed describing the function.  In these notes
we will discuss some basic forms of random hashing that are very
efficient to implement, and yet have sufficient randomness for some
important applications. 
\subsection{Definition and properties}	
	\begin{definition}
	A hash function $h:U\to [m]$ is a random variable in the class of all functions $U\to [m]$, that is, it consists of a random variable $h(x)$ for each $x\in U$.
\end{definition}

For a hash function, we care about roughly three things:
\begin{description}
	\item[Space] The size of the random seed that is necessary to calculate $h(x)$ given $x$,
	\item[Speed] The time it takes to calculate $h(x)$ given $x$,
	\item[Properties of the random variable.]
\end{description}

In the next sections we will mention different desirable properties of
the random hash functions, and how to implement them them efficiently.
First we introduce universal hashing in
Section \ref{sec:universal}, then we introduce strongly universal
hashing in Section \ref{sec:strong-universal}. In both cases, we
present very efficient hash function if the keys are 32- or 64-bit
integers and the hash values are bit strings.  In Section
\ref{sec:ranges} we show how we can efficiently produce hash values in
arbitrary integer ranges. In Section \ref{sec:strings}, we
show how to hash keys that are strings. Finally, in 
Section~\ref{sec:beyond-strong}, we briefly mention hash functions
that have stronger properties than strong universality.

\section{Universal hashing}\label{sec:universal}
The concept of universal hashing was introduced by Carter and Wegman in \cite{CW79}. 
We wish to generate 
a random hash function $h:U\fct[m]$ from a key universe $U$ to
a set of hash values $[m]=\{0,\ldots,m-1\}$.
We think of $h$ as a random variable following some distribution over 
functions $U\fct [m]$. 
We want $h$
to be {\em universal\/} which means
that for any given 
distinct keys $x,y\in U$, when $h$ is picked at random (independently of
$x$ and $y$), we have {\em low collision probability}:
\[\Pr_{h}[h(x)= h(y)]\leq 1/m.\]
For many applications, it suffices if for some $c=O(1)$, we have
\[\Pr_{h}[h(x)= h(y)]\leq c/m.\]
Then $h$ is called {\em $c$-approximately universal}.

In this chapter we will first give some concrete applications of
universal hashing. Next we will show how to implement universal
hashing when the key universe is an integer domain
$U=[u]=\{0,\ldots,u-1\}$ where the integers fit in a machine word,
that is, $u\leq 2^w$ where $w=64$ is the word length. In later
chapters we will show how to make efficient universal hashing for large objects
such as vectors and variable length strings.

\begin{exercise}
		Is the truly independent hash function $h: U \to \left[m\right]$ universal?
\end{exercise}
\begin{exercise}
	If a hash function $h: U \to \left[m\right]$ has collision probability $0$, how large must $m$ be?
\end{exercise}
\begin{exercise}
	Let $u\leq m$. Is the identity function $f(x)=x$ a universal hash function $[u]\to [m]$?
\end{exercise}

\subsection{Applications} One of the most classic applications of 
universal hashing is {\em hash tables with chaining}.  We have a set
$S\subseteq U$ of keys that we wish to store so that we can find any
key from $S$ in expected constant time. Let $n=|S|$ and $m\geq n$.  We
now pick a universal hash function $h:U\fct [m]$, and then create an
array $L$ of $m$ lists/chains so that for $i\in [m]$, $L[i]$ is the
list of keys that hash to $i$. Now to find out if a key $x\in U$ is in
$S$, we only have to check if $x$ is in the list $L[h(x)]$.  This
takes time proportional to $1+|L[h(x)]|$ (we add 1 because it takes
constant time to look up the list even if turns out to be empty). 

Assume that $x\not\in S$ and that $h$ is universal. Let $I(y)$ be an indicator variable
which is $1$ if $h(x)=h(y)$ and $0$ otherwise.
Then the expected number of elements in $L[h(x)]$ is
\[E_h[|L[h(x)]|]=E_h\left[\sum_{y\in S}I(y)\right]=
\sum_{y\in S}E_h[I(y)]=\sum_{y\in S}\Pr_h[h(y)=h(x)]\leq n/m\leq 1.\]
The second equality uses {\em linearity of expectation}.
\begin{exercise} 
\begin{itemize}
\item[(a)] What is the 
expected number of elements in $L[h(x)]$ if $x\in S$?
\item[(b)] What bound do you get if $h$ is only 2-approximately universal?
\end{itemize}
\end{exercise}
The idea of hash tables goes back to \cite{Dum56}, and hash tables were the prime
motivation for the introduction of universal hashing in \cite{CW79}. 
For a text book description, see, e.g., \cite[\S 11.2]{CormAlgo}. 

\drop{With universal hashing, 
we can also implement two level
hashing using linear space, supporting lookups in constant time
\cite{FKS84}. For a text book description, see, e.g.,  \cite[\S 11.5]{CormAlgo}.}

A different application is that of assigning a unique {\em signature\/}
$s(x)$ to each key. Thus we want $s(x)\neq s(y)$ for all
distinct keys $x,y\in S$. To get this, we pick a universal hash function 
$s:U\fct [n^3]$. The probability of an error (collision) is calculated as
\[\Pr_s[\exists \{x,y\}\subseteq S: s(x)=s(y) ]\leq 
\sum_{\{x,y\}\subseteq S}\Pr_s[s(x)=s(y)]\leq {n\choose 2}/n^3<1/(2n).\]
The first inequality is a {\em union bound\/}: that the probability of that at least one of multiple events happen is at most the sum of their probabilities.

The idea of signatures is particularly relevant when the keys are
large, e.g., a key could be a whole text document, which then
becomes identified by the small signature. This idea could
also be used in connection with hash tables, letting the list $L[i]$ store
the signatures $s(x)$ of the keys that hash to $i$, that is, $L[i]=\{s(x) | x\in X, h(x) = i\}$. To check if $x$
is in the table we check if $s(x)$ is in $L[h(x)]$. 
\begin{exercise} With $s:U\fct [n^3]$ and $h:U\fct [n]$ independent universal
hash functions, for a given $x\in U\setminus S$, what is the
probability of a {\em false positive\/} when we search $x$, that is, what is
the probability that there is a key $y\in S$ such that
$h(y)=h(x)$ and $s(y)=s(x)$ ? 
\end{exercise}

Below we study implementations of universal hashing.

\subsection{Multiply-mod-prime}
Note that if $m\geq u$, we can just let $h$ be the identity (no
randomness needed) so we may assume that $m<u$. We may also assume that
$m>1$; for if $m=1$, then $[m]=\{0\}$ and then we have the trivial
case where all keys hash to $0$.

The classic universal hash function from  \cite{CW79} is
based on a prime number $p\geq u$. We pick a uniformly random
$a\in [p]_+=\{1,\ldots,p-1\}$ and $b\in [p]=\{0,\ldots,p-1\}$, and define $h_{a,b}:[u]\fct [m]$ by
\begin{equation}\label{eq:def-univ-prime}
h_{a,b}(x)=((ax+b) \bmod p) \bmod m
\end{equation}
Given any distinct $x,y\in [u]\subseteq[p]$, we want to argue that for
random $a$ and $b$ that
\begin{equation}\label{eq:prime}
\Pr_{a\in[p]_+,\;b\in[p]}[h_{a,b}(x)=h_{a,b}(y)]<1/m.
\end{equation}
The strict inequality uses our assumption that $m>1$. Note that with
a truly random hash function into $[m]$, the collision probability is
exactly $1/m$, so we are claiming that $h_{a,b}$ has a strictly better
collision probability.

In most of our proof, we will consider all $a\in [p]$, including
$a=0$. Ruling out $a=0$, will only be used in the end to get the
tight bound from \req{eq:prime}. Ruling out $a=0$ makes sense because
all keys collide when $a=0$.

We need only one basic fact about primes:
\begin{fact}\label{fact:non-zero} If $p$ is prime and
$\alpha,\;\beta\in[p]_+$ then $\alpha\beta \not\equiv 0 \pmod p$.
\end{fact}
Let $x,y \in [p], x\neq y$ be given. For given pair $(a,b)\in [p]^2$, define $(q,r)\in [p]^2$ by
\begin{eqnarray}
(ax+ b) \bmod p& =& q  \label{eq:x-prime}\\
(ay+ b) \bmod p& =& r.\label{eq:y-prime}
\end{eqnarray}
\begin{lemma}\label{lem:1-1} Equations \req{eq:x-prime} and \req{eq:y-prime} define
a 1-1 correspondence between pairs $(a,b)\in [p]^2$ and pairs $(q,r)\in [p]^2$.
\end{lemma}
\begin{proof}
For a given pair $(r,q)\in [p]^2$, we will show that there is at most
one pair $(a,b)\in [p]^2$ satisfying \req{eq:x-prime} and \req{eq:y-prime}.
Subtracting 
\req{eq:x-prime} from \req{eq:y-prime} modulo $p$, we get
\begin{equation}\label{eq:prime-diff}
(ay+b)-(ax+ b) \equiv a(y-x) \equiv r-q\pmod p,
\end{equation}
We claim that there
is at most one $a$ satisfying \req{eq:prime-diff}. 
Suppose
there is another $a'$ satisfying \req{eq:prime-diff}.
Subtracting the equations with $a$ and $a'$, we get
\[(a-a')(y-x) \equiv 0 \ppmod p ,\]
but since $a-a'$ and $y-x$ are
both non-zero modulo $p$, this contradicts Fact \ref{fact:non-zero}. There
is thus at most one $a$ satisfying \req{eq:prime-diff} for
given $(q,r)$. With this $a$, we need $b$ to satisfy \req{eq:x-prime}, 
and this determines $b$ as
\begin{equation}
b=(q-ax) \bmod p.
\end{equation}
Thus, for each pair 
$(q,r)\in [p]^2$, there is at most one pair $(a,b)\in[p]^2$ satisfying \req{eq:x-prime} and \req{eq:y-prime}. On the other
hand, \req{eq:x-prime} and \req{eq:y-prime} define
a unique pair $(q,r)\in [p]^2$ for each pair $(a,b)\in[p]^2$.
We have
$p^2$ pairs of each kind, so the correspondence must be 1-1.
\end{proof}
Since $x\neq y$, by Fact \ref{fact:non-zero}, 
\begin{equation}\label{eq:zero}
r=q\iff a=0.
\end{equation}
Thus, when we pick $(a,b)\in [p]_+\times [p]$, we
get $r\neq q$.

Returning to the proof of \req{eq:prime}, we get a collision
$h_{a,b}(x)=h_{a,b}(y)$ if and only if $q\bmod m = r\bmod m$. Let us
fix $q$ and set $i= q\bmod m$. There are at most 
$\ceil{p/m}$ values $r\in [p]$ with $r\bmod m=i$ and  one of them
is $r=q$. Therefore, the number of $r\in[p]\setminus \{q\}$ with
$r\bmod m=i=q\bmod m$ is at most $\ceil{p/m}-1\leq (p+m-1)/m-1=(p-1)/m$.
However, there are values of $i'\in[m]$ with only $\floor{p/m}$ 
values of $q'\in [p]$ with $q'\bmod m=i'$, and then
the number of $r'\in[p]\setminus \{q'\}$ with
$r'\bmod m=i'=q'\bmod m$ is $\floor{p/m}-1<\ceil{p/m}-1$.
Summing over all $p$ values of $q\in [p]$, 
we get that the number of $r\in[p]\setminus \{q\}$ 
with $r\bmod m=i=q\bmod m$ is strictly less than $p(\ceil{p/m}-1)\leq 
p(p-1)/m$. Then our 1-1 correspondence implies that there are strictly less than $p(p-1)/m$ collision pairs $(a,b)\in [p]_+\times [p]$. Since each
of the $p(p-1)$ pairs from $[p]_+\times [p]$ are equally likely, we
conclude that the collision probability is strictly below $1/m$, as claimed
in \req{eq:prime}.

\begin{exercise} Suppose we for our hash function also 
consider $a=0$, that is, for random $(a,b)\in
[p]^2$,  we define
the hash function $h_{a,b}:[p]\fct[m]$ by 
\[h_{a,b}(x)=((ax+b)\bmod p)\bmod m.\]
\begin{itemize}
\item[(a)] Show that this function may not be universal.
\item[(b)] Prove that it is always 2-approximately universal, that is, for
any distinct $x,y\in [p]$,
\[\Pr_{(a,b)\in[p]^2}[h_{a,b}(x)=h_{a,b}(y)]<  2/m.\]
\end{itemize}
\end{exercise}

\subsubsection{Implementation for 64-bit keys}\label{sec:Mersenne-impl}
Let us now consider the implementation of our hashing scheme 
\[h(x)=((ax+b) \bmod p) \bmod m\]
for the typical case of 64-bit keys in a standard imperative
programming language such as C. Let's say the hash values are $20$ bits,
so we have $u=2^{64}$ and $m=2^{20}$.

Since $p>u=2^{64}$, we generally need
to reserve more than 64 bits for $a\in[p]_+$, so the product $ax$
has more than 128 bits. To compute $ax$, we now have the issue that
multiplication of $w$-bit numbers automatically discards overflow,
returning only the $w$ least significant bits of the product. However,
we can get the product of 32-bit numbers, representing them as 64-bit numbers, 
and getting the full 64-bit result. We need at least 6 such 64-bit 
multiplications to compute $ax$.

Next issue is, how do we compute $ax \bmod p$? For 64-bit numbers we
have a general mod-operation, though it is rather slow, and here we
have more than 128 bits.

An idea from \cite{CW79} is to let $p$ be a Mersenne prime, that is, a prime
of the form $2^q-1$. Useful examples of such Mersenne primes are $2^{61}-1$ and
$2^{89}-1$.
The point in using a Mersenne prime $p=2^q-1$ is that
\begin{equation}\label{eq:Mersenne}
x\equiv x\bmod 2^{q}+\floor{x/2^q}\ppmod {p}.
\end{equation}
\begin{exercise} Prove that \req{eq:Mersenne} holds. Hint: 
Argue that $x \bmod 2^q + \floor{x/2^q}=x-\floor{x/2^q}p$.
\end{exercise}
Using \req{eq:Mersenne} gives us the following
C-code to compute $y=x\bmod p$:
\begin{verbatim}
y=(x&p)+(x>>q);
if (y>=p) y-=p;
\end{verbatim}
\begin{exercise} Argue that the above code sets $y=x\bmod p$ assuming that 
$x$, $y$, $p$, and $q$ are represented in the same unsigned integer type, 
and that $x<2^{2q}$. In particular, argue we can apply the above code
if $x=x_1x_2$ where $x_1,x_2\in [p]$.
\end{exercise}

\begin{exercise} Assuming a language like C supporting 64-bit multiplication (discarding overflow beyond 64 bits), addition, shifts and bit-wise Boolean operations, but no general
mod-operation, sketch the code to compute $((ax+b)\bmod p)\bmod m$
with $p=2^{89}-1$ and $m=2^{20}$. You should assume that both your
input and your output is arrays of unsigned 32-bit numbers, most
significant number first.
\end{exercise}

\subsection{Multiply-shift}\label{sec:mult-shift-univ}
We shall now turn to a truly practical universal hashing scheme proposed
by Dietzfelbinger et al.~\cite{DHKP97}, yet ignored by most text books. It generally addresses hashing
from $w$-bit integers to $\ell$-bit integers. We pick a uniformly random odd
$w$-bit integer $a$, and then we compute $h_a:[2^w]\fct[2^\ell]$,
as
\begin{equation}\label{eq:def-univ-shift}
h_a(x)=\floor{(ax\bmod 2^w)/2^{w-\ell}}
\end{equation}
This scheme gains an order of magnitude in speed over the scheme from \req{eq:def-univ-prime}, exploiting operations that are fast on standard computers.
Numbers are stored as bit strings, with the least significant bit to the
right. Integer division by a power of two is thus accomplished by a 
right shift. For hashing 64-bit integers, we further exploit
that $64$-bit multiplication
automatically discards overflow, which is the same as multiplying 
modulo $2^{64}$. Thus, with $w=64$, we end up with the following
C-code:
\begin{verbatim}
#include <stdint.h> //defines uint64_t as unsigned 64-bit integer.
uint64_t hash(uint64_t x, uint64_t l, uint64_t a) {
  // hashes x universally into l<=64 bits using random odd seed a. 
  return (a*x) >> (64-l);
}
\end{verbatim}
\drop{As a convenient notation, for any
bit-string $z$ and integers $j>i\geq 0$, $z[i,j)=z[i,j-1]$ denotes the
  number represented by bits $i,\ldots,j-1$ (bit $0$ is the least
  significant bit, which confusingly, happens to be rightmost in the standard representation), so
\[z[i,j)=\floor{(z\bmod 2^j)/2^i}\textnormal{ and }h_a(x)=(ax)[w-\ell,w).\]}
This scheme is many times faster and simpler to implement than
the standard multiply-mod-prime scheme, but the analysis is more subtle.

It is convenient to think of the bits of a number as indexed
with bit 0 the least significant bit. The scheme is simply
extracting bits $w-\ell,\ldots,w-1$ from the product $ax$, as illustrated below.
\begin{center}
\leavevmode
\epsfig{file=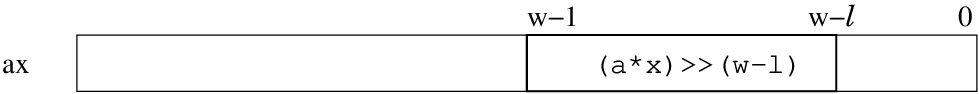,width=5in}
\end{center}
We will prove that multiply-shift is 2-approximately universal, that is, for $x\neq y$,
\begin{equation}\label{eq:mult-shift}
\Pr_{\textnormal{odd }a\in[2^w]}[h_a(x)= h_a(y)]\leq 2/2^\ell=2/m.
\end{equation}
We have $h_a(x)= h_a(y)$ if and only if $ax$ and $ay=ax+a(y-x)$ agree
on bits $w-\ell,\ldots,w-1$. This match requires that bits
$w-\ell,\ldots,w-1$ of $a(y-x)$ are either all $\tto$s or all $\ttl$s.
More precisely, if we get no carry from bits $0,\ldots,w-\ell$ when we
add $a(y-x)$ to $ax$, then $h_a(x)= h_a(y)$ exactly when bits
$w-\ell,\ldots,w-1$ of $a(y-x)$ are all $\tto$s. On the other hand, if we
get a carry $\ttl$ from bits $0,\ldots,w-\ell$ when we add $a(y-x)$ to $ax$,
then $h_a(x)= h_a(y)$ exactly when bits $w-\ell,\ldots,w-1$ of $a(y-x)$ are
all $\ttl$s.
To prove \req{eq:mult-shift}, it thus suffices to prove that
the probability that bits
$w-\ell,\ldots,w-1$ of $a(y-x)$ are all $\tto$s or all $\ttl$s 
is at most $2/2^{\ell}$.

We will exploit that any odd number $z$ is relatively prime
to any power of two:
\begin{fact}\label{fact:non-zero-relative}
If $\alpha$ is odd and $\beta\in [2^q]_+$ then $\alpha\beta \not\equiv 0 \pmod {2^q}$.
\end{fact}
Define $b$ such that $a=1+2b$. Then $b$ is uniformly distributed in
$[2^{w-1}]$.
Moreover, define $z$ to be the odd number satisfying $(y-x)=z 2^i $. Then
\[a(y-x)=z 2^i+bz 2^{i+1}.\]
\begin{center}
\leavevmode
\epsfig{file=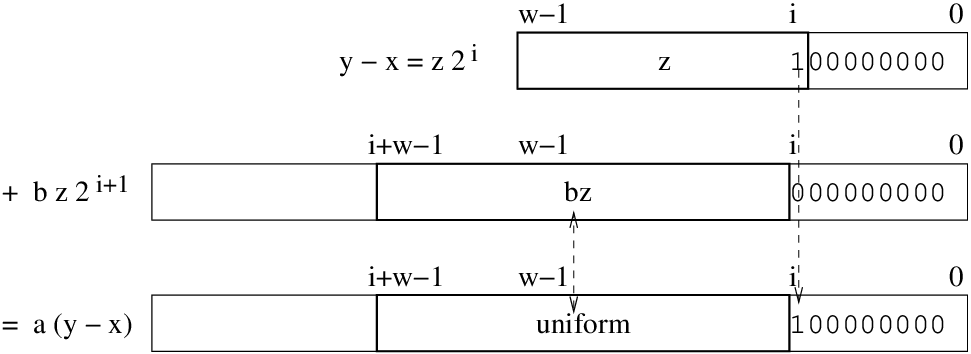,width=5in}
\end{center}
Now, we prove that $bz \bmod 2^{w-1}$ must be uniformly distributed in $[2^{w-1}]$. First, note that there is 
a 1-1 correspondence between the $b\in [2^{w-1}]$ and the products
$bz \bmod 2^{w-1}$; for if there
were another $b'\in [2^{w-1}]$ with $b'z\equiv bz\pmod{ 2^{w-1}}\iff
z(b'-b)\equiv 0\pmod{ 2^{w-1}}$, then this
would contradict Fact \ref{fact:non-zero-relative} since $z$ is odd. But then the uniform distribution on $b$ implies that
$bz \bmod 2^{w-1}$ is uniformly distributed. We conclude
that $a(y-x)=z 2^i+bz 2^{i+1}$ has \tto{} in bits $0,\ldots,i-1$, \ttl{} in bit
$i$, and a uniform distribution on bits $i+1,\ldots,i+w-1$. 

We have a collision $h_a(x)=h_a(y)$ if $ax$ and $ay=ax+a(y-x)$ are
identical on bits $w-\ell,\ldots,w-1$. The two are always different in bit
$i$, so if $i\geq w-\ell$, we have $h_a(x)\neq h_a(y)$ regardless of
$a$. However, if $i<w-\ell$, then because of carries, we could have
$h_a(x)=h_a(y)$ if bits $w-\ell,\ldots,w-1$ of $a(y-x)$ are either all \tto{}s, or all
\ttl{}s. Because of the uniform distribution, either event happens
with probability $1/2^\ell$, for a combined probability bounded by $2/2^\ell$.
This completes the proof of \req{eq:mult-shift}.
\begin{exercise} Why is it important that $a$ is odd?  Hint: consider the case where $x$ and $y$ differ only in the most significant bit. 
\end{exercise}
\begin{exercise}
Does there exist a key $x$ such that $h_a(x)$ is the same regardless of the random odd number $a$?
	If so, can you come up with a real-life application where this is a disadvantage?
\end{exercise}
\section{Strong universality}\label{sec:strong-universal}
We will now consider {\em strong universality} \cite{CW81}. For
$h:[u]\fct[m]$, we consider {\em pairwise events} of the form that
for given distinct keys $x,y\in[u]$ and possibly
non-distinct hash values $q,r\in[m]$, we have $h(x)=q$ and
$h(y)=r$. We say a random hash function $h:[u]\fct[m]$ is {\em
  strongly universal\/} if the probability of every pairwise event is
$1/m^2$. We note that if $h$ is strongly universal, it is also universal
since 
\[\Pr[h(x)=h(y)]=\sum_{q\in[m]}\Pr[h(x)=q\,\wedge
  \,h(y)=q]=m/m^2=1/m.\]
\drop{From Lemma \ref{lem:1-1}, we immediately get
\begin{observation}
For prime $p$ and $(a,b)$ uniform in $[p]^2$, the hash function
$h_{a,b}:[p]\to[p]$ defined by $h_{a,b}(x)=(ax+b)\bmod p$ is strongly
universal.
\end{observation}
We note the following equivalent formulation of strong universality:}
\begin{observation}
An equivalent definition of strong universality is that each key is hashed
uniformly into $[m]$, and that every two distinct keys are hashed independently.
\end{observation}
\begin{proof}
First assume strong universality and consider distinct keys $x,y\in
U$.  For any hash value $q\in [m]$, $\Pr[h(x)=q]=\sum_{r\in[m]}
\Pr[h(x)=q\,\wedge \,h(y)=r]=m/m^2=1/m$, so $h(x)$ is uniform in
$[m]$, and the same holds for $h(y)$. Moreover, for any hash value $r\in[m]$,
\begin{align*}
\Pr[h(x)=q\mid h(y)=r]&=\Pr[h(x)=q\,\wedge
  \,h(y)=r]/\Pr[h(y)=r]\\
&=(1/m^2)/(1/m)=1/m=\Pr[h(x)=q]\textnormal,
\end{align*}
so $h(x)$ is independent of $h(y)$. For the converse direction, when $h(x)$ and $h(y)$ are independent, 
$\Pr[h(x)=q\,\wedge \,h(y)=r]=\Pr[h(x)=q]\cdot\Pr[h(y)=r]$, and 
when $h(x)$ and $h(y)$ are uniform, $\Pr[h(x)=q]=\Pr[h(y)=r]=1/m$,
so $\Pr[h(x)=q]\cdot\Pr[h(y)=r]=1/m^2$.
\end{proof}
Emphasizing the independence, strong universality is also called {\em 2-independence}, as it concerns a pair of \emph{two} events.

\begin{exercise}
	Generalize $2$-independence. What is $3$-independence? $k$-independence?
\end{exercise}

As for universality, we may accept some relaxed notion of strong universality.
\begin{definition}
We say a random hash function $h:U\fct[m]$ is {\em $c$-approximately strongly universal\/} if
\begin{enumerate}
\item $h$ is \emph{$c$-approximately uniform}, meaning for every $x\in U$ and for every hash value $q\in[m]$, we have $\Pr[h(x)=q]\leq c/m$
  and
\item every pair of distinct keys hash independently.
\end{enumerate}
\end{definition}
\begin{exercise}
If $h$ is $c$-approximately strongly universal, what is an upper bound on the 
pairwise event probability, $$\Pr[h(x)=q\land h(y)=r]?$$
\end{exercise}
\begin{exercise} Argue that if $h:U\fct [m]$ is $c$-approximately strongly universal, then 
$h$ is also $c$-approximately universal.
\end{exercise} 
\begin{exercise}
	Is Multiply-Shift $c$-approximately strongly universal for any constant $c$?
\end{exercise}
\subsection{Applications}
One very important application of strongly universal hashing is
{\em coordinated sampling}, which is crucial to the handling of Big Data and
machine learning. The basic idea is that we based on small samples can
reason about the similarity of huge sets, e.g., how much they
have in common, or how different they are.

First we consider sampling from a single set $A\subseteq U$ using
a strongly universal hash function $h:U\fct [m]$ and a threshold
$t\in \{0,\ldots,m\}$. We now sample $x$ if $h(x)<t$, which by uniformity 
happens with probability
$t/m$ for any $x$.  Let $S_{h,t}(A)=\{x\in A\mid h(x)<t\}$ denote the
resulting sample from $A$. Then, by linearity of expectation, 
$E[|S_{h,t}(A)|]=|A|\cdot t/m$. Conversely, this means that
if we have $S_{h,t}(A)$, then we can estimate $|A|$ as
$|S_{h,t}(A)|\cdot m/t$.

We note that the universality from Section \ref{sec:universal} does
not in general suffice for any kind of sampling. If we, for example,
take the multiplication-shift scheme from Section
\ref{sec:mult-shift-univ}, then we always have $h(0)=0$, so $0$ will
always be sampled if we sample anything, that is, if $t>0$. 

The important application is, however, not the sampling from a single
set, but rather the sampling from different sets $B$ and $C$ so that
we can later reason about the similarity, estimating the sizes of their union $B\cup C$ and intersection
$B\cap C$. 

Suppose we for two different sets $B$ and $C$ have found the samples
$S_{h,t}(B)$ and $S_{h,t}(C)$. Based on these we can compute the
sample of the union as the union of the samples, that is,
$S_{h,t}(B\cup C)=S_{h,t}(B)\cup S_{h,t}(C)$. Likewise, we can compute
the sample of the intersection as $S_{h,t}(B\cap C)=S_{h,t}(B)\cap
S_{h,t}(C)$. We can then estimate the size of the union and intersection
multiplying the corresponding sample sizes by $m/t$.

The crucial point here is that the sampling from different sets can be
done in a distributed fashion as long as a fixed $h$ and $t$ is
shared, {\em coordinating the sampling\/} at all locations. This is used,
e.g., in machine learning, where we can store the samples of many
different large sets. When a new set comes, we sample it, and compare
the sample with the stored samples to estimate which other set it has most in common
with.
Another cool application of coordinated sampling is on the Internet where all routers can store samples of
the packets passing through \cite{DuffieldG08}. If a packet is sampled, it is sampled by all
routers that it passes, and this means that we can follow the packets route through the network. If the routers did not use coordinated sampling, the chance that the same packet would be sampled at multiple routers would be very small.

\begin{exercise}
	Given $S_{h,t}(B)$ and $S_{h,t}(C)$, how would you estimate
the size of the symmetric difference $(B\setminus C)\cup (C\setminus B)$?
\end{exercise}

Below, in our mathematical reasoning, we only talk about the sample
$S_{h,t}(A)$ from a single set $A$. However, as described above, in many applications,
$A$ represent a union $B\cup C$ or intersection $B\cap C$ of
different sets $B$ and $C$.

To get a fixed sampling probability $t/m$ for each $x\in U$, we only
need that $h:U\to[m]$ is uniform. This ensures that the estimate
$|S_{h,t}(A)|\cdot m/t$ of $|A|$ is unbiased, that is,
$\E[|S_{h,t}(A)|\cdot m/t]=|A|$.  The reason that we also want the pairwise independence of strong 
universality is that we want $|S_{h,t}(A)|$ to be concentrated around
its mean $|A|\cdot t/m$ so that we can trust the estimate
$|S_{h,t}(A)|\cdot m/t$ of $|A|$.

For $a\in A$, let $X_a=[h(a)<t]$, $X=\sum_{a\in A}X_a$, and $\mu=\E[X]$.
Then $X=|S_{h,t}(A)|$, but the reasoning below applies when $X_a$ is
{\em any\/} 0-1 indicator variable that depends only $h(a)$ (in this context, $t$ is
just a constant).

Because $h$ is strongly universal, for any distinct $a,b\in A$, we
have that $h(a)$ and $h(b)$ are independent, and hence so are $X_a$ and $X_b$. Therefore
$X=\sum_{a\in A}X_a$ is the sum of pairwise independent 0-1 variables. Now the following
concentration bound applies to $X$.
\begin{lemma} Let $X=\sum_{a\in A} X_a$ where the $X_a$ are pairwise independent 0-1 variables.  Let $\mu=\E[X]$. Then $\Var(X)\leq\mu$ and for any $q>0$,
\begin{equation}\label{eq:sqrt-mu}
\Pr[|X-\mu|\geq q\sqrt \mu]\leq 1/q^2.
\end{equation}
\end{lemma}
\begin{proof}
  For $a\in A$, let $p_a=\Pr[X_a]$. Then $E[X_a]=p_a$ and $\Var[X_a]=p_a(1-p_a)\leq p_a=\E[X_a]$. The variance of a sum of pairwise independent variables is the sum of their variances,
  so
  \[\Var[X]=\sum_{a\in A}\Var[X_a]\leq \sum_{a\in A}\E[X_a]=\mu.\]
By definition, the standard
deviation of $X$ is $\sigma=\sqrt{\Var[X]}$, and by Chebyshev's
inequality (see, e.g., \cite[Theorem 3.3]{MR95}), for any $q>0$,
\begin{equation}
\Pr[|X-\mu|\geq q\sigma]\leq 1/q^2.
\end{equation}
This implies \req{eq:sqrt-mu} since $\sigma\leq\sqrt\mu$.
\end{proof}
\begin{exercise} Suppose that $|A|=100,\!000,\!000$ and $p=t/m=1/100$. Then $E[X]=\mu=1,\!000,\!000$. Give an upper bound for the
probability that $|X-\mu|\geq 10,\!000$. These numbers correspond to a 1\% 
sampling rate and a 1\% error.
\end{exercise}
The bound from \req{eq:sqrt-mu} is good for predicting range of outcomes, but often what we have is an experiment giving us a concrete value for our random variable $X$, and now
we want some confidence interval for the unknown mean $\mu$ that we are trying
to estimate.
\begin{lemma}\label{lem:confidence} Let $X$ be a random variable and $\mu=\E[X]$. Suppose 
  \req{eq:sqrt-mu} holds, that is, $\Pr[|X-\mu|\geq q\sqrt \mu]\leq
  1/q^2$ for any given $q$. Then for any given error probability $P$,
  the following holds with probability at least $1-P$,
\begin{equation}\label{eq:confidence}
  X-\sqrt{2X/P}\ <\ \mu\ <\  \max\{8/P,\ X+\sqrt{4X/P}\}.
\end{equation}
\end{lemma}
\begin{proof}
We will show that each of the two inequalities fail with probability at most $P/2$.
First we address the lower-bound, which is the simplest.
From \req{eq:sqrt-mu} with $q=\sqrt{2/P}$, we get that 
\[\Pr[X\geq \mu+\sqrt{2\mu/P}]\leq P/2.\]
However,
\[\mu\leq X-\sqrt{2X/P}\implies \mu\leq X- \sqrt{2\mu/P} \iff X\geq \mu+\sqrt{2\mu/P}\textnormal,.\]
so we conclude that 
\[\Pr[\mu\leq X-\sqrt{2X/P}]\leq \Pr[X\geq \mu+\sqrt{2\mu/P}]\leq P/2.\]
We now address the upper-bound in \req{eq:confidence}. 
From \req{eq:sqrt-mu} with $q=\sqrt{2/P}$, we get
\[\Pr[X\leq \mu-\sqrt{2\mu/P}]\leq P/2.\]
Suppose $\mu\geq 8/P$. Then $\sqrt{2\mu/P}\leq \mu/2$, so $X\leq \mu/2$ implies
$X\leq \mu-\sqrt{2\mu/P}$. However, $X>\mu/2$ and
$\mu \geq X+2\sqrt{X/P}$ implies $\mu \geq X+2\sqrt{\mu/(2P)}=X+\sqrt{2\mu/P}$, hence
$X\leq \mu-\sqrt{2\mu/P}$. Thus we conclude that $\mu\geq\max\{8/P,\ X+2\sqrt{X/P}\}$
implies $X\leq \mu-\sqrt{2\mu/P)}$, hence that
\[\Pr[\mu\geq\max\{8/P,\ X+2\sqrt{X/P}\}]\leq \Pr[X\leq \mu-\sqrt{2\mu/P}]\leq P/2.\]
This completes the proof that \req{eq:confidence} is satisfied with probability $1-P$.
\end{proof}
\begin{exercise}
In science we often want confidence $1-P=95\%$. Suppose we run an experiment
yielding $X=1000$ in Lemma \ref{lem:confidence}. What confidence interval do you get for the underlying mean?
\end{exercise}

\subsection{Multiply-mod-prime}\label{sec:strong-multiply-mod-prime}
The classic strongly universal hashing scheme is a multiply-mod-prime scheme.
For some prime $p$, uniformly at random we pick $(a,b)\in[p]^2$ and define
$h_{a,b}:[p]\fct[p]$ by 
\begin{equation}\label{eq:strong-univ}
h_{a,b}(x)=(ax+b)\bmod p.
\end{equation}
To see that this is strongly universal, consider distinct keys 
$x,y\in[p]$ and possibly
non-distinct hash values $q,r\in[p]$, $h_{a,b}(x)=q$ and $h_{a,b}(x)=r$.
This is exactly as in \req{eq:x-prime} and \req{eq:y-prime}, and
by Lemma \ref{lem:1-1}, we have a 1-1 correspondence between 
pairs $(a,b)\in[p]\times[p]$ and pairs $(q,r)\in [p]^2$.
Since $(a,b)$ is uniform in $[p]^2$ it follows that $(q,r)$ is
uniform in $[p]^2$, hence that the pairwise event
$h_{a,b}(x)=q$ and $h_{a,b}(x)=r$ happens with probability $1/p^2$.
\begin{exercise}\label{ex:prime-strong} For prime $p$, let $m,u\in[p]$.
    For uniformly random $(a,b)\in [p]^2$,  define
the hash function $h_{a,b}:[u]\fct[m]$ by 
\[h_{a,b}(x)=((ax+b)\bmod p)\bmod m.\]
The $\bmod\, m$ operation preserves the pairwise independence of hash values.
\begin{itemize}
\item[(a)] Argue for any $x\in[p]$ and $q\in[m]$ that 
  \begin{equation}\label{eq:prime-uniform}
  (1-m/p)/m< \Pr[h_{a,b}(x)=q]<(1+m/p)/m.
  \end{equation}
  In particular, it follows that $h_{a,b}$ is 2-approximately strongly universal.
\item[(b)] In the universal multiply-mod-prime hashing from Section
\ref{sec:universal}, we insisted on $a\neq 0$, but now we
consider all $a\in[p]$. Why this difference?
\end{itemize}
\end{exercise}
  For a given $u$ and $m$, it follows from \req{eq:prime-uniform} that we can get multiply-mod-prime $h_{a,b}:[u]\fct[m]$ arbitrarily
close to uniform by using a large enough prime $p$. In practice,
we will therefore often think of $h_{a,b}$ as strongly universal, ignoring the
error $m/p$.

\subsection{Multiply-shift}\label{sec:strong-mult-shift}
We now present a simple generalization from \cite{Die96} of the
universal multiply-shift scheme from Section \ref{sec:universal} that
yields strong universality. As a convenient notation, for any
bit-string $z$ and integers $j>i\geq 0$, $z[i,j)=z[i,j-1]$ denotes the
  number represented by bits $i,\ldots,j-1$ (bit $0$ is the least
  significant bit, which confusingly, happens to be rightmost in the standard representation), so
\[z[i,j)=\floor{(z\bmod 2^j)/2^i}.\]
To get strongly universal hashing $[2^w]\fct [2^\ell]$, we may pick any 
$\ol w\geq w+\ell -1$. For any pair 
$(a,b)\in [\ol w]^2$, we define $h_{a,b}:[2^w]\fct [2^\ell]$ by
\begin{equation}\label{eq:shift-strong}
h_{a,b}(x)=(ax+b)[\ol w-\ell,\ol w).
\end{equation}
As for the universal multiply shift, we note that the scheme of \req{eq:shift-strong} is easy to implement with convenient parameter choices, e.g., with
$\ol w=64$, $w=32$ and $\ell=20$, we get the C-code:
\begin{verbatim}
#include <stdint.h> 
// defines uint32/64_t as unsigned 32/64-bit integer.
uint32_t hash(uint32_t x, uint32_t l, uint64_t a, uint64_t b) {
  // hashes 32-bit x strongly universally into l<=32 bits 
  // using the random seeds a and b. 
  return (a*x+b) >> (64-l);
}
\end{verbatim}
The above code uses 64-bit multiplication like in Section \ref{sec:mult-shift-univ}.
However, in  Section \ref{sec:mult-shift-univ}, we got universal hashing
from 64-bit keys to up to 64-bit hash values. Here we get
strongly universal hashing from 32-bit keys to up to 32-bit hash values.
For strongly universal hashing of 64-bit keys, we can use the
pair-multiply-shift that will be introduced in Section \ref{sec:pair-mult-shift}, and to get up to 64-bit hash values, we can use the concatenation of
hash values that will be introduced in Section \ref{sec:large-range}.
Alternatively, if we have access to fast 128-bit multiplication, then
we can use it to hash directly from 64-bit keys to 64-bit hash values.

We will now prove that the scheme from \req{eq:shift-strong} is strongly
universal. In the proof we will reason a lot about uniformly
distributed variables, e.g., if $X\in [m]$ is uniformly distributed
and $\beta$ is a constant integer, then $(X+\beta)\bmod m$ is also
uniformly distributed in $[m]$.  More interestingly, we have
\begin{fact}\label{fct:*a} Consider two positive integers $\alpha$ and $m$ 
that are relatively prime, that is, $\alpha$ and $m$ have no common prime
factor. If $X$ is uniform
in $[m]$, then $(\alpha X)\bmod m$ is also uniformly distributed 
in $[m]$. Important cases are (a) if $\alpha<m$ and $m$ is prime, and (b) if
$\alpha$ is odd and $m$ is a power of two.
\end{fact}
\begin{proof} 
We want to show that for every $y\in[m]$ there is
at most one $x\in[m]$ such that $(\alpha x)\bmod m=y$, for then
there must be exactly one $x\in [m]$ for each $y\in[m]$, and vice versa.
Suppose we had distinct $x_1,x_2\in[m]$ such that
$(\alpha x_1)\bmod m=y=(\alpha x_2)\bmod m$. Then
$\alpha(x_2-x_1)\bmod m=0$, so $m$ is a divisor of $\alpha(x_2-x_1)$.
By the fundamental theorem of arithmetic, every positive integer
has a unique prime factorization, so all prime factors of $m$ have
to be factors of $\alpha(x_2-x_1)$ in same or higher powers. Since
$m$ and $\alpha$ are relatively prime, no prime factor of $m$ is factor
of $\alpha$, so the prime factors of $m$ must all be factors of $x_2-x_1$ in same or higher powers. Therefore
$m$ must divide $x_2-x_1$, contradicting the assumption $x_1\not\equiv x_2 \pmod m$.
Thus, as desired, for any $y\in [m]$, there is at most one $x\in[m]$ such that
$(\alpha x)\bmod m=y$. 
\end{proof}
\begin{theorem}\label{thm:shift-strong}  When $a,b\in [2^{\ol w}]$ are uniform and independent,
the multiply-shift scheme from \req{eq:shift-strong} is strongly universal.
\end{theorem}
\begin{proof}
Consider any distinct keys $x,y\in [2^w]$. We want to show that
$h_{a,b}(x)$ and $h_{a,b}(y)$ are independent uniformly distributed variables
in $[2^\ell]$.

Let $s$ be the index of the least
significant $\ttl$-bit in $(y-x)$ and let $z$ be the odd number such that
$(y-x)=z2^s$. Since $z$ is odd and $a$ is uniform in $[2^{\ol w}]$, by Fact \ref{fct:*a} (b), we have that
$az$ is uniform in $[2^{\ol w}]$. Now $a(y-x)=az 2^s$ has all $\tto$s in
bits $0,..,s-1$ and a uniform distribution on bits
$s,..,s+\ol w-1$. The latter implies that
$a(y-x)[s,..,\ol w-1]$ is uniformly distributed in $[2^{\ol w-s}]$.

Consider now any fixed value of $a$. Since $b$ is still uniform in
$[2^{\ol w}]$, we get that $(ax+b)[0,\ol w)$ is uniformly distributed, implying
that $(ax+b)[s,\ol w)$ is uniformly distributed. This holds for any fixed value of $a$, so we conclude that $(ax+b)[s,\ol w)$ and 
$a(y-x)[s,\ol w)$ are independent random variables, each uniformly distributed
in $[2^{\ol w-s}]$.

Now, since $a(y-x)[0,s)=0$, we get that 
\[(ay+b)[s,\infty)=((ax+b)+a(y-x))[s,\infty)=(ax+b)[s,\infty)+a(y-x)[s,\infty).\]
The fact that $a(y-x)[s,\ol w)$ is uniformly distributed independently of
$(ax+b)[s,\ol w)$ now implies that $(ay+b)[s,\ol w)$ is uniformly distributed
independently of $(ax+b)[s,\ol w)$. However, $\ol w\geq w+\ell-1$ and $s<w$
so $s\leq w-1\leq \ol w-\ell$. Therefore 
$h_{a,b}(x)=(ax+b)[\ol w-\ell,\ol w)$ and 
$h_{a,b}(y)=(ay+b)[\ol w-\ell,\ol w)$ are independent uniformly 
distributed variables in $[2^\ell]$.
\end{proof}
In order to reuse the above proof in more complicated settings, we
crystallize a technical lemma from the last part:
\begin{lemma}\label{lem:shift-strong}  Let $\ol w\geq w+\ell -1$. 
Consider a random function $g:U\fct[2^{\ol w}]$  with the property that there for any distinct
  $x,y\in U$ exists a positive $s<w$, determined by $x$ and $y$ (and
  not by $g$), such that
  $(g(y)-g(x))[0,s)=0$ while $(g(y)-g(x))[s,\ol w)$ is uniformly
      distributed in $[2^{\ol w-s}]$. For $b$ uniform in $[2^{\ol w}]$ and
independent of $g$, 
    define $h_{g,b}:U\fct [2^\ell]$ by
\[h_{g,b}(x)=(g(x)+b)[\ol w-\ell,\ol w).\]
Then $h_{g,b}(x)$ is strongly universal.
\end{lemma}
In the proof of Theorem~\ref{thm:shift-strong}, we would have $U=[2^w]$ and $g(x)=ax[0,\ol w)$, and $s$ was the least significant set bit in $y-x$.

\subsection{Vector multiply-shift}
Our strongly universal multiply shift scheme generalizes nicely to 
vector hashing. The goal is to get strongly universal hashing 
from $[2^w]^d$ to $2^\ell$. With $\ol w\geq w+\ell -1$, we pick
independent uniform $a_0,\ldots,a_{d-1},b\in [2^{\ol w}]$ and
define $h_{a_0,\ldots,a_{d-1},b}:[2^w]^d\fct [2^\ell]$ by
\begin{equation}\label{eq:vector-shift}
h_{a_0,\ldots,a_{d-1},b}(x_0,\ldots,x_{d-1})=
\left(\left(\sum_{i\in [d]}a_ix_i\right)+b\right)[\ol w-\ell,\ol w).
\end{equation}
\begin{theorem}\label{thm:vector-shift}
The vector multiply-shift scheme from \req{eq:vector-shift}
is strongly universal.
\end{theorem}
\begin{proof}
We will use Lemma \ref{lem:shift-strong} to prove that this scheme is strongly universal. We define
$g:[2^w]^d\fct [2^{\ol w}]$ by
\[g(x_0,\ldots,x_{d-1})=\left(\sum_{i\in [d]}a_ix_i\right)[0,\ol w).\]
Consider two distinct keys $x=(x_0,\ldots,x_{d-1})$ and $y=(y_0,\ldots,y_{d-1})$.
Let $j$ be an index such that $x_j\neq y_j$ and such that the index $s$ of
the least significant set bit is as small as possible. Thus $y_j-x_j$ has $\ttl$ in bit $s$, and all $i\in[d]$ have
$(y_j-x_j)[0,s)=0$. As required by Lemma \ref{lem:shift-strong}, $s$ is determined from the keys only, as required by Lemma \ref{lem:shift-strong}. Then 
\[(g(y)-g(x))[0,s)=\left(\sum_{i\in [d]}a_i(y_i-x_i)\right)[0,s)=0\]
regardless of $a_0,\ldots,a_{d-1}$. Next we need to show that
$(g(y)-g(x))[s,\ol w)$ is uniformly
      distributed in $[2^{\ol w-s}]$. The trick is to first fix
all $a_i$, $i\neq j$, arbitrarily, and then argue that $(g(y)-g(x))[s,\ol w)$
is uniform when $a_i$ is uniform in $[2^{\ol w}]$. Let $z$ be the
odd number such that $z2^s=y_j-x_j$. Also, let $\Delta$ be the constant defined by
\[\Delta 2^s=\sum_{i\in [d],i\neq j}a_i(y_i-x_j).\]
Now 
\[g(y)-g(x)=(a_jz+\Delta)2^s.\]
With $z$ odd and $\Delta$ a fixed constant, the uniform distribution
on $a_j\in[2^{\ol w}]$ implies that $(a_jz+\Delta)\bmod 2^{\ol w}$ is uniform
in $[2^{\ol w}]$ but then $(a_jz+\Delta)\bmod 2^{\ol w-s}=(g(y)-g(x))[s,\ol w)$
is also uniform in $[2^{\ol w-s}]$. Now Lemma \ref{lem:shift-strong}
implies that the vector multiply-shift scheme from \req{eq:vector-shift}
is strongly universal.
\end{proof}
\begin{exercise}\label{ex:naive-strong}
Corresponding to the universal hashing from Section \ref{sec:universal},
suppose we tried with $\ol w=w$ and just used 
random odd $a_0,\ldots,a_{d-1}\in [2^w]$ and a random
$b\in [2^w]$,
and defined
\[h_{a_0,\ldots,a_{d-1},b}(x_0,\ldots,x_{d-1})=
\left(\left(\sum_{i\in [d]}a_ix_i\right)+b\right)[w-\ell, w).\]
Give an instance showing that this simplified vector hashing scheme is not
remotely universal.
\end{exercise}
Our vector hashing can also be used for universality, 
where it gives collision probability $1/2^{\ell}$. As a small
tuning, we could skip adding $b$, but then we would only get
the same $2/2^\ell$ bound as we had in Section \ref{sec:universal}.

\subsection{Pair-multiply-shift}\label{sec:pair-mult-shift}
A cute trick from \cite{BHKKR99} allows us roughly
double the speed of vector hashing, the point being that multiplication is
by far the slowest operation involved. We will use exactly the
same parameters and seeds as for \req{eq:vector-shift}.
However, assuming that the dimension $d$ is even, we replace
\req{eq:vector-shift} by 
\begin{equation}\label{eq:vector-pair-shift}
h_{a_0,\ldots,a_{d-1},b}(x_0,\ldots,x_{d-1})=
\left(\left(\sum_{i\in [d/2]}(a_{2i}+x_{2i+1})(a_{2i+1}+x_{2i})
\right)+b\right)[\ol w-\ell,\ol w).
\end{equation}
This scheme handles pairs of coordinates $(2i,2i+1)$ with a single multiplication. Thus, with $\ol w=64$ and $w=32$, we handle
each pair of $32$-bit keys with a single $64$-bit multiplication.
\begin{exercise} [a bit more challenging]\label{ex:pair-multiply-shift}
Prove that the scheme defined by \req{eq:vector-pair-shift} is
strongly universal. One option is to prove a tricky
generalization of Lemma \ref{lem:shift-strong} where $(g(y)-g(x))[0,s)$
  may not be $0$ but can be any deterministic function of $x$ and
  $y$. With this generalization, you can make a proof similar to that for Theorem
  \ref{thm:vector-shift} with the same definition of $j$ and $s$.
\end{exercise}
Above we have assumed that $d$ is even. In particular this is a case,
if we want to hash an array of 64-bit integers, but cast it as an 
array of 32-bit numbers. If $d$ is odd, we can use the pair-multiplication for
the first $\floor{d/2}$ pairs, and then just add $a_dx_d$ to the sum.

\paragraph{Strongly universal hashing of 64-bit keys to 32 bits}
For the in practice quite important case where we want strongly
 universal hashing of 64-bit keys to at most 32 bits, 
we can use the following tuned code:
\begin{verbatim}
#include <stdint.h> 
// defines uint32/64_t as unsigned 32/64-bit integer.
uint32_t hash(uint64_t x, uint32_t l, 
              uint64_t a1, uint64_t a2, uint64_t b) {
  // hashes 64-bit x strongly universally into l<=32 bits 
  // using the random seeds a1, a2, and b. 
  return ((a1+x)*(a2+(x>>32))+b) >> (64-l);
}
\end{verbatim}
The proof that this is indeed strongly universal
is very similar to the one used for Exercise \ref{ex:pair-multiply-shift}.

\section{Fast hashing to arbitrary ranges}\label{sec:ranges}
Using variants of multiply-shift, we have shown very efficient methods
for hashing into $\ell$-bit hash values for $\ell\leq 32$, but what if
we want hash values in $[m]$ for any given $m<2^\ell$.

The general problem is if we have a good hash function 
$h:U\rightarrow [M]$, and now we want hash values in
$[m]$ where $m<M$. What we need is a function $r:[M]\fct[m]$ that is
{\em most uniform\/} in the sense that for any $z\in [m]$, the
number of $y\in [M]$ that map to $z$ is either $\floor{M/m}$ or $\ceil{M/m}$.
\begin{exercise} 
Prove that if $h$ is $c$-approximately strongly universal and $r$ is most uniform, then
$r\circ h$, mapping $x$ to $r(h(x))$ is $(1+m/M)c$-approximately strongly universal.
\end{exercise}
An example of a most uniform function $r$ is $y\mapsto 
y\bmod m$. We already used this $r$ in Exercise \ref{ex:prime-strong} where
we first had a [1-approximately] strongly universal hash function into $[p]$ and then
applied $\bmod m$. However, computing $\bmod m$ is much more expensive than 
a multiplication on most computers unless $m$ is a power-of-two. An alternative way to get
a most uniform hash function $r:[M]\fct[m]$ is to set
\begin{equation}\label{eq:range-div}
r(y)=\floor{ym/M}.
\end{equation}
\begin{exercise} 
Prove that $r$ defined in \req{eq:range-div} is most uniform.
\end{exercise}
While \req{eq:range-div} is not fast in general, it is very fast if
$M=2^\ell$ is a power-of-two, for then the division is just a right
shift, and then \req{eq:range-div} is computed by \texttt{(y*m)>>l}.
One detail to note here is that the product $ym$ has to be computed in
full with no discarded overflow, e.g., if $y$ and $m$ are 32-bit integers,
we need
64-bit multiplication. Combining this with the code for strongly
universal multiply-shift from Section \ref{sec:strong-mult-shift}, we
hash a 32-bit integer $x$ to a number in $[m]$, $m<2^{32}$, using the
C-code:
\begin{verbatim}
#include <stdint.h> 
// defines uint32/64_t as unsigned 32/64-bit integer.
uint32_t hash(uint32_t x, uint32_t m, uint64_t a, uint64_t b) {
  // hashes x strongly universally into the range [m]
  // using the random seeds a and b.
  return (((a*x+b)>>32)*m)>>32;
}
\end{verbatim}
Above, $x$ and $m$ are 32-bit integers while $a$ and $b$ are uniformly 
random 64-bit 
integers. We note that all the above calculations are done with
64-bit integers since they all involve 64-bit operands. As required
in Section \ref{sec:strong-mult-shift}, we automatically discard the 
overflow beyond 64 bits from \texttt{a*x}. However
\texttt{(a*x+b)>>32} only uses the 32 least significant bits,
so multiplied with the 32-bit integer \texttt m, we get
the exact product in 64 bits with no overflow. From the above
Exercises, it immediately follows that he C-code function above is
a $2$-approximately strongly universal hash function from 32-bit integers to integers
in $[m]$.

\subsection{Hashing to larger ranges}\label{sec:large-range}
So far, we have been focused on hashing to $32$-bit numbers or less.
If we want larger hash values, the most efficient method is often just
to use multiple hash functions and concatenate the output. The idea is
captured by the following exercise.
\begin{exercise}\label{ex:combine-range}
Let $h_0$ be a $c_0$-approximately strongly universal hash function from $U$ to $R_0$
and $h_1$ be a $c_1$-approximately strongly universal hash function from $U$ to $R_1$.
Define the combined hash function $h:U\fct R_0\times R_1$ by 
\[h(x)=(h_0,h_1).\]
Prove that $h$ is $(c_0c_1)$-approximately strongly universal.
\end{exercise}
A simple application of Exercise \ref{ex:combine-range} is if
$h_0$ and $h_1$ are the strongly universal hash functions from Section \ref{sec:pair-mult-shift}, generating 32-bit hash values from 64-bit keys. 
Then the combined hash  function $h$ is a strongly universal 
function from 64-bit keys to 64-bit hash values. It is the
fastest such hash function known, and it uses only two 64-bit multiplications.

\section{String hashing}\label{sec:strings}

\subsection{Hashing vector prefixes}\label{sec:var-max}
Sometimes what we really want is to hash vectors of length up to $D$
but perhaps smaller. As in the multiply-shift hashing schemes, we assume
that each coordinate is from $[2^w]$. The simple point is that we only want
to spend time proportional to the actual length $d\leq D$. With $\ol w\geq w+\ell -1$, we pick
independent uniform $a_0,\ldots,a_{D-1}\in [2^{\ol w}]$. For
even $d$, we define
$h:\bigcup_{\textnormal{even\ } d\,\leq \,D}[2^w]^d\fct [2^\ell]$ by
\begin{equation}\label{eq:vector-pair-shift-var}
h_{a_0,\ldots,a_D}(x_0,\ldots,x_{d-1})=
\left(\left(\sum_{i\in [d/2]}(a_{2i}+x_{2i+1})(a_{2i+1}+x_{2i})
\right)+a_d\right)[\ol w-\ell,\ol w).
\end{equation}
\begin{exercise}
Prove that the above even prefix version of pair-multiply-shift
is strongly universal. In the proof you may assume that the
original pair-multiply-shift from \req{eq:vector-pair-shift} is strongly universal, as you may have proved in 
Exercise \ref{ex:pair-multiply-shift}. Thus we
are considering two vectors $x=(x_0,\ldots,x_{d-1})$ and 
$y=(y_0,\ldots,y_{d'-1})$. You should consider both the case
$d'=d$ and $d'\neq d$.
\end{exercise}

\subsection{Hashing bounded length strings}\label{sec:bounded-length-strings}
Suppose now that we want to hash strings of 8-bit characters, e.g.,
these could be the words in a book. Then the nil-character is
not used in any of the strings. Suppose that we only want to
handle strings up to some maximal length, say, 256. 

With the prefix-pair-multiply-shift scheme from \req{eq:vector-pair-shift-var},
we have a very fast way of hashing strings of $d$ 64-bit integers,
casting them as $2d$ 32-bit integers.
A simple trick now is to allocate a single array $x$ of $256/8=32$ 64-bit integers.
When we want to hash a string $s$ with $c$ characters, we first
set $d=\lceil c/8\rceil$ (done fast by \texttt{d=(c+7)>>3}). Next
we set $x_{d-1}=0$, and finally we do a memory copy of $s$ into $x$ (using
a statement like \texttt{memcpy(x,s,c)}).
Finally, we apply \req{eq:vector-pair-shift-var} to $x$.

Note that we use the same variable array $x$ every time we want to hash a string $s$.
Let $s^*$ be the image of $s$ created as a $c^*=\lceil c/8\rceil$ 
length prefix of $x$.
\begin{exercise}\label{ex:bounded-length-strings} Prove that if $s$ and $t$ are two strings of length
at most 256, neither containing the nil-character, then their images
$s^*$ and $t^*$ are different. Conclude that we now have a strongly 
universal hash functions for such strings.
\end{exercise}
\begin{exercise}\label{ex:book} Implement the above hash function for strings. Use it in a chaining based hash table, and apply it
to count the number of distinct words in a text (take any pdf-file and convert it to ASCII, e.g., using \texttt{pdf2txt}).

To get the random numbers defining your hash functions, you can go to \texttt{random.org}.

One issue to consider when you implement a hash table is that you want
the number $m$ of entries in the hash array to be as big as the
number of elements (distinct words), which in our case is not known in
advance. Using a hash table of some start size $m$, you can maintain a
count of the distinct words seen so far, and then double the size
when the count reaches, say, $m/2$.

Many ideas can be explored for optimization, e.g., if we are willing
to accept a small false-positive probability, we can replace each word
with a 32- or 64-bit hash value, saying that a word is new only if it
has a new hash value.

Experiment with some different texts: different languages, and different lengths. What happens with the vocabulary?

The idea now is to check how
much time is spent on the actual hashing, as compared with the real
code that both does the hashing and follows the chains in the hash
array. However, if we just compute the hash values, and don't use them, then
some optimizing compilers, will notice, and just do nothing. You should therefore add up all the hash values, and output the result, just to force the compiler to do the computation.
\end{exercise}

\drop{
\paragraph{Linear Probing}
As an alternative in Exercise \ref{ex:book}, you could implement the
hash table using linear probing. However, with linear probing, if you want the
hashing to be safe in the sense of getting expected constant time for
any input, then this does require some more powerful hash
functions (which we have not covered yet). Currently, the most efficient schemes combine methods from
\cite{PT12:charhash,Tho09:string-hash}. Referring to
\cite[\S 1, \S 6]{PT12:charhash} for the definition of tabulation hashing, you can
apply it directly on words with up to 4 characters, that is,
32-bits. If the words are longer, you can first hash them, as
described above, down to 32-bits, and then rehash them using simple
tabulation. The reason that this works is explained in 
\cite{Tho09:string-hash}. In
\cite[\S 6]{Tho09:string-hash} you will also see that it
is possible to store the distinct strings
directly in the hash array.

}

\subsection{Hashing variable length strings}\label{sec:var}
We now consider the hashing of a variable length string
$x_0x_1\cdots x_d$ where all characters belong to some domain
$[u]$.

We use the method from \cite{DGMP92}, which first picks a prime $p\geq u$.
The idea is to view  $x_0,\ldots,x_d$ as coefficients of a degree $d$ 
polynomial 
\[P_{x_0,\ldots,x_d}(\alpha)=\sum_{i=0}^d x_i\alpha^i\bmod p\]
over $\mathbb Z_p$. As seed for our hash function, we
pick an argument $a\in [p]$, and
compute the hash function
\[h_a(x_0\cdots x_d)=P_{x_0,\ldots,x_d}(a).\]
Consider some other string 
$y=y_0y_1\cdots y_{d'}$, $d'\leq d$. We claim that
\[\Pr_{a\in[p]}[h_a(x_0\cdots x_d)=h_a(y_0\cdots y_{d'})]\leq d/p\]
The proof is very simple. By definition, the collision happens
only if $a$ is root in the polynomial 
$P_{y_0,\ldots,y_{d'}}-P_{x_0,\ldots,x_d}$. Since the strings are different,
this polynomial is not the constant zero. Moreover, its degree is at most
$d$. Since the degree is at most $d$, the fundamental theorem of algebra
tells us that it has at most $d$ distinct roots, and the probability that a random $a\in[p]$ is among these roots is at most $d/p$.

Now, for a fast implementation using Horner's rule, it is better to reverse the order of the coefficients, and instead use the polynomial 
\[P_{x_0,\ldots,x_d}(a)=\sum_{i=0}^d x_{d-i}a^i\bmod p\]
Then we compute $P_{x_0,\ldots,x_d}(a)$ using the recurrence
\begin{itemize}
\item $H^0_a=x_0$
\item $H^i_a=(a H^{i-1}_a+x_i) \bmod p$
\item $P_{x_0,\ldots,x_d}(a)=H^d_a$.
\end{itemize}
With this recurrence, we can easily update the hash value if new
character $x_{d+1}$ is added to the end of the string $x_{d+1}$.
It only takes an addition and a multiplication modulo $p$. For speed, we would let $p$ be a Mersenne prime, e.g. $2^{89}-1$. 


The collision probability $d/p$ may seem fairly large, but assume that
we only want hash values in the range $m\leq p/d$, e.g, for
$m=2^{32}$ and $p=2^{89}-1$, this would allow for strings of length up to
$2^{57}$, which is big enough for most practical purposes. Then
it suffices to compose the string hashing with a universal hash function
from $[p]$ to $[m]$. Composing with the previous multiply-mod-prime scheme, we
end up using three random seeds $a,b,c\in [p]$, and then
compute the hash function as
\[h_{a,b,c}(x_0,\ldots,x_d)=\left(\left(a\left(\sum_{i=0}^d x_{d-i} c^i\right)+b\right)\bmod p\right)\bmod m.\]
\begin{exercise}\label{ex:var-2univ} Consider two strings $\vec x$ and $\vec y$,
each of length at most $p/m$. Argue that the 
the collision probability $\Pr[h_{a,b,c}(\vec x)=h(\vec y)]\leq 2/m$.
Thus conclude that $h_{a,b,c}$ is a 2-approximately universal hash function mapping
strings of length at most $p/m$ to $[m]$.
\end{exercise}
Above we can let $u$ be any value bounded by $p$. With $p=2^{89}-1$, we
could use $u=2^{64}$ thus dividing the string into $64$-bit characters.
\begin{exercise} Implement the above scheme and run it to get a $32$-bit
signature of a book. 
\end{exercise}
\paragraph*{Major speed-up}
The above code is slow because of the multiplications modulo Mersenne primes, one for every 64 bits in the string.

An idea for a major speed up is to divide you string into chunks $X_0,\ldots,X_j$
of 32 integers of 64 bits, the last chunk possibly being shorter. We want a single universal 
hash function 
$r:\bigcup_{d\leq 32}[2^{64}]^d\fct[2^{64}]$. A good choice would be
to use our strongly universal pair-multiply-shift scheme from \req{eq:vector-pair-shift-var}. It only outputs 32-bit numbers, but if we use two different
such functions, we can concatenate their hash values in a single 64-bit
number.
\begin{exercise} Prove that if $r$ has collision probability $P$, and
if $(X_0,\ldots,X_j)\neq(Y_0,\ldots,Y_{j'})$, then
\[\Pr[(r(X_0),\ldots,r(X_j))=(r(Y_0),\ldots,r(Y_{j'}))]\leq P.\]
\end{exercise}
The point above is that in the above is that $r(X_0),\ldots,r(X_j)$
is 32 times shorter than $X_0,\ldots,X_j$. We can now apply our slow
variable length hashing based on Mersenne primes to the
reduced string $r(X_0),\ldots,r(X_j)$. This only adds $P$ to the overall collision probability.
\begin{exercise}
Implement the above tuning. How much faster is your hashing now?
\end{exercise}

\paragraph{Splitting between short and long strings} When 
implementing a generic string hashing code, we do not know in advance 
if it is going to be applied mostly to short or to long strings.
Our scheme for bounded length strings from Section
\ref{sec:bounded-length-strings} is faster then the generic scheme presented
above for variable length strings. In practice it is a good idea to 
implement both: have the scheme from Section
\ref{sec:bounded-length-strings} implemented for strings of length up to 
some $d$, e.g., $d$ could be 32 64-bit integers as in the above blocks, and
then only apply the generic scheme if the length is above $d$.

\paragraph{Major open problem} Can we get something simple and fast 
like multiply-shift to work directly for strings, so that we do not
need to compute polynomials over prime fields?

\section{Beyond strong universality}\label{sec:beyond-strong}
In this note, we have focused on universal and strongly universal
hashing. However, there are more advanced algorithmic
applications that need more powerful hash functions. This lead Carter
and Wegman \cite{CW81} to introduce $k$-independent hash functions. A
hash random function $H: U \to [m]$ is {\em $k$-independent\/} if for
any distinct keys $x_1, \dots, x_k \in [u]$, the hash values $H(x_1),
\dots, H(x_k)$ are independent random variables, each uniformly
distributed in $[m]$. In this terminology, 2-independence is the same
as strongly universal. For prime $p$, we can implement a
$k$-independent $H:[p]\to[p]$ using $k$ random coefficients
$a_0,\ldots,a_{k-1}\in [p]$, defining
\[H(x)=\sum_{i=0}^{k-1}a_i x^i \bmod p.\]
However, there is no efficient implementation of $k$-independent
hashing on complex objects such as variable length strings.  What we
can do instead is to use the signature idea from Section
\ref{sec:universal}, and first hash the complex keys to
unique hash values from a limited integer domain $[u]$.
\begin{exercise}
Let $h:U\to[m]$ map $S\subseteq U$ collision free to $[u]$ and let
$H:[u]\to[m]$ be $k$-independent. Argue that $H\circ h:U\to [m]$,
mapping $x$ to $H(h(x))$, is $k$-independent when restricted to keys
in $S$.
\end{exercise}
Next we have to design the hash function $h$ so that it is collision 
free with high enough probability.
\begin{exercise}
Suppose we know that we are only going to deal with a (yet unknown)
set $S$ of at most $n$ strings, each of length at most $n$, set the
parameters $p$ and $m$ of $h_{a,b,c}$ in Exercise \ref{ex:var-2univ}
so that the probability that we get any collision between strings in
$S$ is at most $1/n^2$.
\end{exercise}
The idea of using hashing to map a complex domain down to a
polynomially sized integer universe, hoping that it is collision free, is
referred to as {\em universe reduction}. This explains why we for more
complex hash function can often assume that the universe is polynomially
bounded. 

We note that $k$-independent hash functions become quite slow when $k$ is
large. Often a more efficient alternative is the tabulation based 
method surveyed in \cite{Tho17:tab-hashr}.

\section*{Acknowledgment} I would like to thank Eva Rotenberg for coming with many good comments and proposals for the text, including several exercises.


\end{document}